\documentclass[10pt,final,journal,twocolumn,twoside,letterpaper]{IEEEtran}%\usepackage{txfonts}  %different fonts
\usepackage{cite}
\usepackage[dvips]{graphicx}
\usepackage{amsmath}
\usepackage{cases}
\usepackage{times}
\usepackage{latexsym}
\usepackage{graphicx}
\usepackage{bm}
\usepackage{amssymb}
\usepackage{stfloats}
\usepackage{cases}
\usepackage{array}
\usepackage{setspace}
\usepackage{fancyhdr}
\usepackage{ulem}
\usepackage{mathtools}
\usepackage{flushend}
\usepackage{xcolor}
\usepackage[center]{caption2}
\usepackage{algorithm}
\usepackage{algorithmic}

\allowdisplaybreaks[4]

\newcaptionstyle{mystyle1}{%
  \centering TABLE  \captiontext \par}
\captionstyle{mystyle1}
\newcaptionstyle{mystyle2}{%
  \captionlabel $.$ \: \captiontext \par}
\captionstyle{mystyle2}
\newcaptionstyle{mystyle3}{%
 \centering \captionlabel $.$ \: \captiontext \par}
\captionstyle{mystyle3}

\renewcommand{\QED}{\QEDopen}

\newtheorem{theorem}{Theorem}

\title{Energy Efficient Joint Resource Allocation and Power Control for D2D Communications}

\author{\authorblockN{Yanxiang~Jiang,~\IEEEmembership{Member,~IEEE}, Qiang~Liu, Fuchun~Zheng,~\IEEEmembership{Senior Member,~IEEE}, Xiqi~Gao,~\IEEEmembership{{Fellow},~IEEE}, and Xiaohu~You,~\IEEEmembership{Fellow,~IEEE}}
\thanks{\small {Copyright (c) 2015 IEEE. Personal use of this material is permitted. However, permission to use this material for any other purposes must be obtained from the IEEE by sending a request to pubs-permissions@ieee.org. }}
\thanks{\small {Manuscript received April 3, 2014; revised November 25, 2014, and May 11, 2015; accepted August 22, 2015.}
This work was supported in part by the National Basic Research Program of China (973 Program) under grant 2012CB316004,  the National High-Tech 863
Program under grant 2015AA01A709, the Natural Science Foundation of China (NSFC) under grant 61221002, and the National Key Project under grant 2012ZX03003013-004.
The review of this paper was coordinated by Dr. F. Ban.}
\thanks{\small Y. Jiang, X. Gao, and X. You are with the National Mobile Communications Research Laboratory,
Southeast University, Nanjing 210096, China (e-mail: \{yxjiang, xqgao, xhyu\}@seu.edu.cn).}
\thanks{\small Q. Liu is with the Tangshan Branch, Hebei Co., Ltd., China Mobile Communications Group (e-mail: {liuqiang1\_ts}@he.chinamobile.com).}
\thanks{\small F. Zheng is with the School of Systems Engineering, University of
Reading, Reading, RG6 6AY, UK (e-mail: f.zheng@reading.ac.uk).}
}

\begin{document}
\maketitle

\begin{abstract}
In this paper, joint resource allocation and power control for energy efficient device-to-device (D2D) communications underlaying cellular networks are investigated.
The resource and power are optimized for maximization of the energy efficiency (EE) of D2D communications.
Exploiting the properties of  fractional programming, we transform the original  nonconvex optimization problem in fractional form into
an equivalent optimization problem in subtractive form. Then, an efficient iterative resource allocation and power control scheme is proposed.
In each iteration, part of the constraints of the EE optimization problem is removed by exploiting the penalty function approach.
{We further propose a novel two-layer approach which allows to find the optimum at each iteration by decoupling the EE optimization problem of
joint resource allocation and power control into two separate steps.}
In the first layer, the optimal power values are obtained by solving a series of maximization problems through root-finding with or without considering
the loss of cellular users' rates.
In the second layer, {the formulated optimization problem belongs to a classical resource allocation problem with
single allocation format which admits a network flow formulation so that it can be solved to optimality.}
Simulation results demonstrate the remarkable improvements in terms of EE by using the proposed iterative resource allocation and power control scheme.
\end{abstract}


%\newpage
\section{Introduction}
Device-to-device (D2D) communications underlaying cellular networks have the potential of increasing spectrum efficiency (SE) and energy efficiency (EE) as well as allowing new peer-to-peer services by taking advantage of the so called proximity and reuse gains \cite{Doppler}. Because D2D users share the same spectrum with cellular users (CUs), sophisticated resource allocation and power control need to be performed to achieve high SE and EE.

Many works, so far, have been carried out on resource allocation to improve throughput and reduce the interference between D2D and cellular connections {\cite{Wang, Min, Xu}.} The authors in \cite{Wang}  proposed to limit the minimum distance between CUs and D2D users reusing the same resources to mitigate the interference from cellular transmission to D2D links. A $\delta_D$-interference limited area control scheme was proposed to manage the interference from CUs to D2D users in \cite{Min}. For the schemes in \cite{Wang} and \cite{Min}, the transmit power of D2D users is assumed to be fixed. The authors in \cite{Xu} optimized the system throughput by introducing a sequential second price auction as the resource allocation mechanism, and much channel information need to be exchanged between D2D users and the base station (BS).The above research mainly focuses on improving the performance by using resource allocation.

Apart from resource allocation, power control is also a key technique to achieve high performance {\cite{Yu, Yu1,Fodor, Xiao, FodorDella}.}
The authors in \cite{Yu} applied a simple power control method to D2D communications, which constrains the signal to interference-plus-noise ratio (SINR) degradation of the cellular link to a certain level. In \cite{Yu1}, the optimum power allocation for different resource sharing modes was derived. In \cite{Fodor}, the authors devised a new distributed power control algorithm, which iteratively determines the SINR targets and allocates transmit power so that the overall power consumption is minimized subject to a sum-rate constraint. For the algorithm in \cite{Fodor}, a near optimal solution was obtained which requires the slow changing path loss and shadowing matrix knowledge at each transmitter. In \cite{Xiao}, a heuristic scheme  was proposed to minimize the downlink transmit power constrained by the demands of users' quality of service {(QoS)}. {In \cite{FodorDella}, the authors investigated the performance of various power control strategies applicable to D2D communications in long term evolution (LTE) networks, and gained valuable insight by quantifying the performance with respect to a utility function maximization approach.} The above research mainly focuses on improving the performance of cellular link or D2D link by using power control.

\pubidadjcol

{Many wireless systems of higher capacity are generally designed to improve SE \cite{ShuYou}.
While the power consumption of mobile devices increases to meet the increasing demand from multimedia applications.}
In contrast, the improvement in battery technology has been very slow, leading to an exponentially increasing gap between the required battery capacity and the available battery capacity.
Therefore, energy efficient system designs which adopt EE as the performance metric have recently draw much attention in both industry
and academia.
A power-efficient mode selection and power allocation scheme was proposed in \cite{Jung}, which was performed based on the exhaustive search of all possible mode combinations of the D2D users.
In \cite{Qiu}, energy efficient power control schemes in three different resource sharing modes were discussed under the maximum transmission power constraint.
{In \cite{Belleschi}, the authors proposed a global optimization model that takes into account mode selection, resource assignment and power allocation in a unified framework with
the aim of minimizing the overall power consumption for D2D communications underlaying a cellular infrastructure.}
{The authors in \cite{Meshkati} presented a number of noncooperative power control games for resource allocation in code-division multiple-access (CDMA) networks
with emphasis on EE, where the users can choose their uplink receivers, transmission rates, and carrier allocation strategies in addition to choosing their transmit powers.
In \cite{Miao}, low-complexity closed-form power control and resource allocation schemes were developed by using a time-averaged bits-per-Joule metric for uplink orthogonal frequency
division multiple access (OFDMA) systems in frequency-selective channels.}
To the best of our knowledge, for the existing research on D2D communications, there
is no study about joint resource allocation and power control for maximization of the EE.

Motivated by the aforementioned reviews, we formulate the resource allocation and power control problem for energy efficient D2D
communications as a nonconvex optimization problem. Exploiting the properties of fractional programming, we transform the considered nonconvex optimization
problem in fractional form into an equivalent optimization problem in subtractive form with a tractable  solution,  which
can be obtained with an iterative approach.
In each iteration, part of the constraints of the equivalent optimization problem is removed by exploiting the
penalty function approach. Then, a two-layer resource allocation and power control scheme is developed to maximize the EE.
The proposed scheme converges fast to the optimal solution in our considered simulation scenario.

The rest of the paper is organized as follows. In Section II, the system model of D2D communications underlaying cellular networks is presented and the corresponding EE optimization problem is formulated. In Section III, the joint resource allocation and power control scheme is developed. Simulation results are shown in Section IV. Final conclusions are drawn in Section V.

\section{System Model and Problem Formulation}

In this paper, the D2D communications operate as an underlay of the general
cellular network which enables local services by a D2D radio. The general cellular
network may be LTE network. The D2D operation itself can be fully transparent to
the users, and the D2D devices are the general user equipments in the cellular
network.
A single cell scenario with multiple CUs and D2D pairs is considered as illustrated in Figure 1. Assume that the uplink resources of  CUs are allowed to be reused by  D2D pairs, and that all the CUs and all the D2D users are randomly distributed in the cell \cite{Xu1}.

\begin{figure}[!t]
\centering %\vspace*{135pt}
\includegraphics[width=0.4\textwidth]{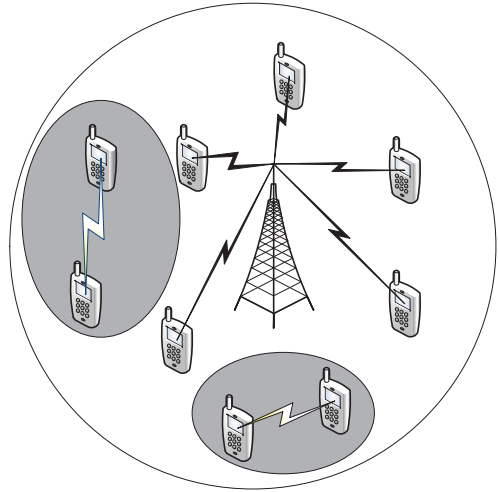}
\captionstyle{mystyle2}
\caption{Schematic of the considered D2D single cell.}
\label{fig1}
\end{figure}

Let $\mathbb{L}=\{1,2,...,N_c\}$, $\mathbb{M}=\{1,2,...,M\}$, and $\mathbb{N}=\{1,2,...,N_d\}$  denote the sets of the CUs, the resource blocks (RBs), and the D2D pairs  in the considered single cell, respectively. $N_c$, $M$, and $N_d$ should satisfy $N_d \le N_c = M$.
Here, one RB refers to one time-frequency resource block which includes 1 time slots  and 12 subcarriers.
Assume that resource allocation of the CUs has already been accomplished, and that the $j$-th RB is allocated to the $k$-th CU. Let {$P_{k,B}$} denote the average received power at the BS from {the $k$-th CU.} Then, it can be expressed as \cite{Min}
\begin{equation}
{\ P_{k,B} = \kappa P_{k,j}(d_{k,B})^{-\chi}}
\end{equation}
where $P_{k,j}$ denotes the transmit power of the $k$-th CU on the $j$-th RB, $d_{k,B}$ denotes the distance between the $k$-th CU and the BS, $\kappa$ and $\chi$ denote the path loss constant and  path loss exponent, respectively.

Assume that the $j$-th RB allocated to the $k$-th CU is reused by the $i$-th D2D pair. Let $P_{i,j}$ denote the transmit power of the transmitter of the $i$-th D2D pair on the $j$-th RB. Let $R(P_{i,j})$ denote the data rate of the $i$-th D2D pair on the $j$-th RB. Then, $R(P_{i,j})$ can be expressed as follows
\begin{equation}
R(P_{i,j})=W\log_{2}(1+\frac{P_{i,j}h_{i,i}}{P_{k,j}h_{k,i}+N_0})
\end{equation}
where $W$ denotes the bandwidth of the considered RB, $h_{i,i}$ denotes the channel gain from the transmitter to the receiver of the $i$-th D2D pair, $h_{k,i}$ denotes the interference channel gain from the $k$-th CU to the receiver of the $i$-th D2D pair, and $N_0$ denotes the noise power.

In order to achieve the objective of green communications and to transmit more bits on unit Joule under the constraint of minimum rate requirement, we use the EE of all the D2D users as the optimization objective. The same choice can be found in the literature, for example, the references \cite{Ng,Xiao1}. To choose the total EE of all the D2D users, we pay more attention to all the D2D users as a whole.
Let $P_C$ denote the circuit power consumption of all the considered D2D pairs, and it is a power offset that is independent of radiated power, but derived from signal processing, and battery backup, etc. Let $U_{EE}$ denote the EE of D2D communications, and it can be defined as follows
\begin{equation}
U_{EE} = \frac{\sum\limits _{i=1} ^{N_d}\sum\limits _{j=1} ^{M}x_{i,j}R(P_{i,j})}{\sum\limits _{i=1} ^{N_d}\sum\limits _{j=1} ^{M}x_{i,j}P_{i,j}+P_C}
\end{equation}
where $x_{i,j}\in\{0,1\}$ indicates whether the $j$-th RB is allocated to the $i$-th D2D pair.
Then, the EE optimization problem can be mathematically formulated as follows
\begin{subequations}\label{eq:parent}
\begin{align}\label{fz}
& \qquad\  \max\limits_{x_{i,j},P_{i,j}}\ U_{EE} \tag{4}\\
\text{s.t.} &\ x_{i,j}\in\{0,1\},\quad\forall{i}\in\mathbb{N},\forall{j}\in\mathbb{M} \\
&\ \sum\limits _{j=1} ^{M}x_{i,j}=1,\quad\forall{i}\in\mathbb{N}\\
&\ \sum\limits _{i=1} ^{N_d}x_{i,j}\leq1,\quad\forall{j}\in\mathbb{M}\\
&\  R(P_{i,j})\geq\gamma_{i},\quad\forall{x_{i,j}=1}\\%[8pt]
&\  P_{k,j}h_{k,i}\leq\tau,\quad\forall{x_{i,j}=1}\\%[8pt]
&\  P_{i,j}\leq P_{max},\quad\forall{x_{i,j}=1}
\end{align}
\end{subequations}
where $\gamma_{i}$ denotes the minimum rate requirement of each D2D pair, $\tau$ denotes the threshold of the allowed interference from the considered CU sharing the same resource with the corresponding D2D pair, and it is set to guarantee that the interference between D2D users and cellular users will not be too serious, $P_{max}$ denotes the maximum transmit power of the transmitter of {each} D2D pair, and we assume that the maximum transmit power of of the transmitter of each D2D pair is the same. Note that a smaller maximum transmit power of D2D pairs may reduce the energy efficiency of D2D pairs and also reduce the interference to the CUs, and that (4b) and (4c) mean that each D2D pair can only reuse one RB and each RB can not be reused by more than one D2D pair.
{Note that all the subcarriers included in one RB are characterized by the same channel gains in our formulated EE optimization problem, and hence no single-user diversity is exploited. However, the assumption of constant channel gain in the allocated subcarriers in one RB makes the power control problem more practical, despite it representing a suboptimal solution.}
{It can also be seen from the above EE optimization problem that there is no constraint on the interference produced towards CUs from D2D users.
On the one hand, in a real scenario, the transmitter and the receiver of each D2D pair are generally located close,
and it is also constrained through (4f) that the maximum transmit power of the transmitter of the $i$-th D2D pair is not larger than $P_{max}$.
Correspondingly, relatively small $P_{i,j}$ is enough to guarantee reliable and efficient communications between the D2D users of each D2D pair.
On the other hand, we have also assumed that the uplink resources of  CUs are reused by D2D pairs.
Correspondingly, the BS will be affected by the interference generated from  D2D users due to $P_{i,j}$. % and $h_{i,B}$.
Generally speaking, the interference level from D2D users tends to be very low due to the low transmit power of D2D transmitters and the distance to the BS.
Also, the BS is normally more interference resistant than  D2D receivers (e.g. due to more advanced interference rejection algorithms).
Therefore, we do not set the constraint on the interference produced towards  CUs from  D2D users.
}

\section{Proposed Joint Resource Allocation and Power Control Scheme}

Given the integer assignment variables $x_{i,j}$, the optimization problem above falls into the scope of combinatorial programming which is NP-hard \cite{Xiao1}.
The fractional objective function in (4), which is a nonconvex function, makes the problem even more complicated.
A brute force approach is generally required to obtain a global optimal solution. However, such a method has an exponential complexity with respect to
(w.r.t.) the number of RBs, and it is computationally impracticable even for a small size system. Therefore, we commit to an effective method to solve
this challenging problem.

\subsection{Problem Equivalence}
We first treat the fractional objective in (4) which can be classified as a nonlinear fractional program \cite{Dinkelbach}. For
description convenience, let $\mathbf{\Phi}$ denote the feasible domain defined by (4a)-(4f).
Let $q^*$ denote the maximum EE of D2D communications. Then, it can be defined as follows
\begin{equation}
q^* = \max\limits_{\{x_{i,j},P_{i,j}\}\in\mathbf{\Phi}}\ \frac{\sum\limits _{i=1} ^{N_d}\sum\limits _{j=1} ^{M}x_{i,j}R(P_{i,j})}{\sum\limits _{i=1} ^{N_d}\sum\limits _{j=1} ^{M}x_{i,j}P_{i,j}+P_C}
\end{equation}
We are now ready to present the following Theorem.
\begin{theorem}
The maximum EE $q^*$ is achieved if and only if
\begin{multline}
\max\limits_{\{x_{i,j},P_{i,j}\}\in\mathbf{\Phi}} \left\{\sum\limits _{i=1} ^{N_d}\sum\limits _{j=1} ^{M}x_{i,j}R(P_{i,j}) \right. \\
\left. -q^*\left(\sum\limits _{i=1} ^{N_d}\sum\limits _{j=1} ^{M}x_{i,j}P_{i,j}+P_C\right)\right\} =0
\end{multline}
\end{theorem}
\begin{proof}
The proof of Theorem 1 is similar to the proof in \cite{Dinkelbach, Xiao1, Ng}.
\end{proof}

It can be seen from Theorem 1 and \cite{Dinkelbach, Xiao1, Ng} that  an equivalent objective function in subtractive form in (6) exists for the considered objective function in fractional form in (5) by exploiting the fractional programming from \cite{Dinkelbach}.
Consequently, we investigate the equivalent objective function in subtractive form in the rest of this paper.

\subsection{Iterative Algorithm for EE Maximization}
By exploiting an iterative algorithm known as the Dinkelbach method \cite{Dinkelbach}, the equivalent objective function in subtractive form in (6) can be solved.
Let $s$ denote the number of iterations, $q_s$ denote the temporary EE,
and $\epsilon$ denote the convergence threshold.
Then, the corresponding iterative resource allocation and power control algorithm can be summarized in Algorithm 1.

\begin{algorithm}[!b]
\caption{Iterative Resource Allocation and Power Control Algorithm}
%\begin{algorithmic}
\begin{itemize}
\item Step 1: Initialization: $s=1$, $q_{s}=0$.
\item Step 2: For the given $q_{s}$, solve the following optimization problem to obtain $x_{i,j}^{'}$ and $P_{i,j}^{'}$
\begin{multline*}
F(q_{s}) =\max\limits_{\{x_{i,j},P_{i,j}\}\in\mathbf{\Phi}}\left\{ \sum\limits _{i=1} ^{N_d}\sum\limits _{j=1} ^{M}x_{i,j}R(P_{i,j}) \right. \\
\left. -q_s\left(\sum\limits _{i=1} ^{N_d}\sum\limits _{j=1} ^{M}x_{i,j}P_{i,j}+P_C\right)\right\}
\end{multline*}
\item Step 3: Set $q_s=\frac{\sum\limits _{i=1} ^{N_d}\sum\limits _{j=1} ^{M}x_{i,j}^{'} R(P_{i,j}^{'})}{\sum\limits _{i=1} ^{N_d}\sum\limits _{j=1} ^{M}x_{i,j}^{'}P_{i,j}^{'}+P_C}$. If $|F(q_{s})|$ $\geq$ $\epsilon$, set $s=s+1$.%[20pt]
\item Step 4: Repeat the steps $2\thicksim3$ until $|F(q_{s})|$ $<$ $\epsilon$.
\end{itemize}
%\end{algorithmic}
\end{algorithm}

%\begin{itemize}
%\item Step
%\end{itemize}

In each iteration, by using the penalty function approach in \cite{Wang1}, the constraints (4d) and (4f) can be removed for the considered optimization problem in step 2 in Algorithm 1.
Assume that each D2D pair can measure the interference power on different RBs to obtain the feasible RB sets that satisfy the constraint (4e).
Let $\mathbf{\Phi^{'}}$ denotes the feasible domain defined by removing the constraints (4d)-(4f).
Then, the corresponding optimization problem can be written in the following equivalent form
\begin{multline}
 \max\limits_{\{x_{i,j},P_{i,j}\}\in\mathbf{\Phi^{'}}} \left\{\sum\limits _{i=1} ^{N_d}\left[\sum\limits _{j=1} ^{M}x_{i,j}R(P_{i,j}) \right. \right. \\
\left.\left. -q_s\sum\limits _{j=1} ^{M}x_{i,j}P_{i,j}\right] +\mathbf{\varphi}_1Z_1+\mathbf{\varphi}_2Z_2\right\}
\end{multline}
where
\begin{equation*}
Z_1=\sum\limits _{i=1} ^{N_d}\sum\limits _{j=1} ^{M}x_{i,j}\min \{0, 1+\frac{P_{i,j}h_{i,i}}{P_{k,j}h_{k,i}+N_0}-2^{\frac{\gamma_{i}}{W}}\}
\end{equation*}
\begin{equation*}
Z_2=\sum\limits _{i=1} ^{N_d}\sum\limits _{j=1} ^{M}x_{i,j}\min \{0,P_{\max}-P_{i,j}\}
\end{equation*}
$\mathbf{\varphi}_1$ and $\mathbf{\varphi}_2$ denote the penalty factors, and they  should be set as large as possible.
Then, (7) can be rewritten as
\begin{multline}
\max\limits_{\{x_{i,j},P_{i,j}\}\in\mathbf{\Phi^{'}}}\sum\limits _{i=1} ^{N_d}\left\{\sum\limits _{j=1} ^{M}x_{i,j}R(P_{i,j})-q_s\sum\limits _{j=1} ^{M}x_{i,j}P_{i,j} \right.\\
\left.+\mathbf{\varphi}_1\sum\limits _{j=1} ^{M}x_{i,j}\min \{0,1+\frac{P_{i,j}h_{i,i}}{P_{k,j}h_{k,i}+N_0}-2^{\frac{\gamma_{i}}{W}}\}\right. \\
\left. +\mathbf{\varphi}_2\sum\limits _{j=1} ^{M}x_{i,j}\min \{0,P_{\max}-P_{i,j}\}\right\}
\end{multline}
Define
\begin{multline*}
f_{i,j}(P_{i,j}) = R(P_{i,j})-q_s P_{i,j} \\
+\mathbf{\varphi}_1\min \{0,1+\frac{P_{i,j}h_{i,i}}{P_{k,j}h_{k,i}+N_0}-2^{\frac{\gamma_{i}}{W}}\} \\
+\mathbf{\varphi}_2\min \{0,P_{\max}-P_{i,j}\}
\end{multline*}
Then, (8) can be expressed as
\begin{equation}
\begin{array}{r@{~}l}
\max\limits_{\{x_{i,j},P_{i,j}\}\in\mathbf{\Phi^{'}}}\ \sum\limits _{i=1} ^{N_d}\sum\limits _{j=1} ^{M}x_{i,j}f_{i,j}(P_{i,j})
\end{array}
\end{equation}

\begin{figure*}[!b]
\hrulefill
\begin{align*}
&a=q_sh_{i,i}h_{i,B}\ln2\tag{17a}\\
&b=q_s\ln2[(2h_{i,i}h_{i,B}N_0+h_{i,i}h_{i,B}P_{k,j}h_{k,j})+h^{2}_{i,B}(P_{k,j}h_{k,i}+N_0)]-h^{2}_{i,B}h_{i,i}\tag{17b}\\
&c=q_s N_0\ln2[h_{i,i}(P_{k,j}h_{k,j}+N_0)+h_{i,B}(P_{k,j}h_{k,i}+N_0)]+q\ln2(P_{k,j}h_{k,i}+N_0)(P_{k,j}h_{k,j}+N_0) %\\
%&  \hspace*{400pt}
-2h_{i,i}h_{i,B}N_0\tag{17c}\\
&d=q_s N_0\ln2(P_{k,j}h_{k,i}+N_0)(P_{k,j}h_{k,j}+N_0)+h_{i,B}P_{k,j}h_{c_j}(P_{k,j}h_{k,i}+N_0)-h_{i,i}N_0(P_{k,j}h_{k,j}+N_0)\tag{17d}
\end{align*}
\end{figure*}

\subsection{The Novel Two-Layer Joint Resource Allocation and Power Control Scheme}

By taking advantage of the constraints (4a)-(4c), we introduce the following theorem.
\begin{theorem}\label{t2}
The optimal solutions of the power values $P_{i,j}$ of the optimization problem in (9) must be the optimal solutions of the optimization
problem $\mathop {\max }\limits_{{P_{i,j}}} {f_{i,j}}({P_{i,j}})$.
\end{theorem}
\begin{proof}
By exploiting the constraints (4a) and (4b), the optimal solutions of the power values $P_{i,j}$ of the optimization problem in (9) must be the optimal solutions of the following equivalent optimization problem
\begin{equation}
\max\limits_{P_{i,j(i)}}\ \sum\limits _{i=1} ^{N_d}f_{i,j(i)}(P_{i,j(i)})
\end{equation}
where $j(i) \in \mathbb{M}$. According to the constraint (4c), if $i \ne i^{'}$, then $j(i) \ne j(i^{'})$. Consequently, the optimization problem in
(10) can be further expressed in the following equivalent $N_d$ independent optimization problems
\begin{equation}
\max\limits_{P_{i,j(i)}}\ f_{i,j(i)}(P_{i,j(i)}), \quad \forall i \in \mathbb{N}
\end{equation}
This completes the proof.
\end{proof}

It can be seen from Theorem \ref{t2} that the original two-dimensional nonconvex optimization problem can be transformed
into two separate optimization problems in each iteration.
Therefore, we propose a novel two-layer joint resource allocation and power control scheme to obtain the optimal solutions.
In the first layer, each D2D pair just needs to determine the maximum value of $f_{i,j}(P_{i,j})$ on the candidate RBs constrained by (4e), and
the corresponding power value can then be determined.
In the second layer, in order to maximize the EE and avoid the interference among D2D pairs, each RB will be allocated to one D2D pair
according to the results of the first layer.

In  the above discussions, the loss of CUs' rates caused by the D2D users is not taken into consideration. In the following,
both the case  without considering the loss of CUs' rates and that for considering it will be studied.
Define
\begin{multline}
\tilde R(P_{i,j})=W\log_{2}(1+\frac{P_{i,j}h_{i,i}}{P_{k,j}h_{k,i}+N_0}) \\
-W\left[\log_{2}(1+\frac{P_{k,j} h_{k,j}}{N_0}) \right. \\
\left. -\log_2(1+\frac{P_{k,j} h_{k,j}}{P_{i,j}h_{i,B}+N_0})\right]%\tag{11}
\end{multline}
\begin{multline}
\theta_{i,j}(P_{i,j})=\tilde R(P_{i,j})-q_sP_{i,j} \\
+\mathbf{\varphi}_1\min \{0,1+\frac{P_{i,j}h_{i,i}}{P_{k,j}h_{k,i}+N_0}-2^{\frac{\gamma_{i}}{W}}\} \\
+\mathbf{\varphi}_2\min \{0,P_{\max}-P_{i,j}\} %\tag{12}
\end{multline}
where the second term in the right hand side of (12) denotes the loss of the $k$-th CU's rate, $h_{k,j}$ denotes the channel gain between the $k$-th CU and the BS on the $j$-th RB, and $h_{i,B}$ denotes the interference channel gain between the transmitter of the $i$-th D2D pair and the BS.
Considering the loss of CUs' rates, we just need to replace $R(P_{i,j})$ with $\tilde R(P_{i,j})$, and replace $f_{i,j}(P_{i,j})$ with $\theta_{i,j}(P_{i,j})$.
{We remark here that, although there is no constraint on the interference produced towards   CUs from D2D users due to $P_{i,j}$ and $h_{i,B}$ in the original EE optimization problem as formulated in (\ref{eq:parent}),
it has actually been constrained implicitly through the third term in the right hand side of (12) in the modified EE optimization problem in the case for considering the loss of  CUs' rates.}
{We also remark here that, by substracting the loss of the CU's rate from the rate of the considered D2D pair as shown in (12)
and replacing $R(P_{i,j})$ with $\tilde R(P_{i,j})$ in the EE of D2D communications as defined in (3),
both the influence of the resource to the rate of CUs and that to the EE of D2D users are taken into consideration in the modified EE optimization problem.
}

\subsubsection{ The First Layer (Power Control)}

When the loss of CUs' rates is not considered,
by taking derivative of $f_{i,j}(P_{i,j})$ with respect to $P_{i,j}$ and considering the constraints (4d) and (4f),
{the optimal transmit power $P_{i,j}^{*}$ can be obtained  as follows
\begin{multline}\label{Phiqq}
P_{i,j}^{*} = {\frac{W}{q_s\ln2}-\frac{P_{k,j}h_{k,i}+N_0}{h_{i,i}},} \\
\text{$(2^{\frac{\gamma_{i}}{W}}-1)\frac{P_{k,j}h_{k,i}+N_0}{h_{i,i}}\leq P_{i,j}\leq P_{\max}$}
\end{multline}
Note that if the above solution is less than the leftmost limit of the constraining interval in (\ref{Phiqq}), then it means that the minimum rate requirement cannot be satisfied to achieve the maximum EE.
If the above solution  is greater than the rightmost limit of the constraining interval in (\ref{Phiqq}), then it means that the maximum transmit power requirement cannot be satisfied to achieve the maximum EE.
If the leftmost limit is greater than the rightmost limit of the constraining interval in (\ref{Phiqq}), then it means that the minimum rate requirement cannot be satisfied even the maximum transmit power is used.
}

When the loss of CUs' rates is considered,
taking derivative of $\theta_{i,j}(P_{i,j})$ with respect to $P_{i,j}$ and considering the constraints (4d) and (4f), we have
\begin{multline}
\theta_{i,j}^{'}(P_{i,j}) = \frac{Wh_{i,i}}{\ln2(P_{i,j}h_{i,i}+P_{k,j}h_{k,i}+N_0)} \\
-\frac{WP_{k,j}h_{k,j}h_{i,B}}{\ln2(P_{i,j}h_{i,B}+P_{k,j}h_{k,j}+N_0)(P_{i,j}h_{i,B}+N_0)} \\
-q_s, \quad \text{$(2^{\frac{\gamma_{i}}{W}}-1)\frac{P_{k,j}h_{k,i}+N_0}{h_{i,i}}\leq P_{i,j}\leq P_{\max}$}
\end{multline}
Let $\theta^{'}_{i,j}(P_{i,j})=0$, i.e.,
\begin{multline}
\frac{Wh_{i,i}}{\ln2(P_{i,j}h_{i,i}+P_{k,j}h_{k,i}+N_0)} \\
-\frac{W P_{k,j}h_{k,j}h_{i,B}}{\ln2(P_{i,j}h_{i,B}+P_{k,j}h_{k,j}+N_0)(P_{i,j}h_{i,B}+N_0)} \\
-q_s=0
\end{multline}
Then, the maximum value of $\theta_{i,j}(P_{i,j})$ can be obtained. (16) can be written in the following equivalent cubic equation
\begin{equation}
aP^{3}_{i,j}+bP^{2}_{i,j}+cP_{i,j}+d=0 %\tag{14}
\end{equation}
whose coefficients are shown in (17a)-(17d). Let
\begin{eqnarray*}
p^{N}_{i,j} &=&-b/3a \\
\delta &=&(b^2-3ac)/(9a^2) \\
\lambda^2 &=&3\delta^2 \\
h &=&-2a\delta^3 \\
y_N &=&(2b^3)/(27a^2)-(bc)/(3a)+d
\end{eqnarray*}
According to the value of $y_N$, there are three cases for the solution of (17) \cite{Nickalls}.
If $y_N>h^2$, one real root exists
\begin{eqnarray*}
\alpha_1 &=& p^{N}_{i,j}+(\frac{-y_N+\sqrt{y^2_N-h^2}}{2a})^{\frac{1}{3}}
\end{eqnarray*}
If $y_N = h^2$, two real roots exist
\begin{eqnarray*}
\alpha_2 &=& p^{N}_{i,j}+\delta \\
\beta_2 &=& p^{N}_{i,j}-2\delta
\end{eqnarray*}
If $y_N < h^2$, three real roots exist
\begin{eqnarray*}
\alpha_3 &=& p^{N}_{i,j}+2\delta\cos\rho \\
\beta_3 &=& p^{N}_{i,j}+2\delta\cos(\frac{2\pi}{3}-\rho) \\
\gamma_3 &=& p^{N}_{i,j}+2\delta\cos(\frac{2\pi}{3}+\rho)
\end{eqnarray*}
where $\rho=\frac{1}{3} \arccos(-y_N/h)$. With all the possible candidate optimal power values for the three cases, the optimal transmit power $P^{*}_{i,j}$ can be readily obtained.

\subsubsection{ The Second Layer (Resource Allocation)}
Let $\xi_{i,j}(P_{i,j}^{*})=f_{i,j}(P_{i,j}^{*})$ in the case  without considering the loss of CUs' rates.
Let $\xi_{i,j}(P_{i,j}^{*})=\theta_{i,j}(P_{i,j}^{*})$ in the case for considering the loss of CUs' rates.
Then, the optimization problem in the second layer can be shown as follows
\begin{subequations}
\begin{align}
& \quad  \max\limits_{x_{i,j}} \sum\limits _{i=1} ^{N_d}\sum\limits _{j=1} ^{M}x_{i,j}\xi_{i,j}(P_{i,j}^{*}) \tag{18} \\%\tag{17}
\text{s.t.} &\  x_{i,j}\in\{0,1\},\quad\forall{i}\in\mathbb{N},\forall{j}\in\mathbb{M}\\
&\ \sum\limits _{j=1} ^{M}x_{i,j}=1,\quad\forall{i}\in\mathbb{N}\\%\tag{18b}\\
&\ \sum\limits _{i=1} ^{N_d}x_{i,j}\leq1,\quad\forall{j}\in\mathbb{M}%\tag{18c}
\end{align}
\end{subequations}
{It can be seen that the above optimization problem  is  a classical resource allocation problem with single allocation format,
which admits a network flow formulation so that it can be solved to optimality \cite{AbrardoBelleschi},
and it also belongs to the assignment problem,
which can be solved by Hungarian method or brand and bound method \cite{Wolsey}.
}

Note that our proposed two-layer joint resource allocation and power control scheme can achieve the optimal solution of the original mixed
combinatorial and nonconvex optimization problem, and that the original two-dimensional optimization problem is also transformed to two
separate optimization problems which can be solved with low complexity.
Clearly, a centralized approach has been assumed in this paper.
This would certainly incur the necessary signaling overhead,
and the impact of such overhead will be an interesting issue for future research.

\section{Simulation Results}

In this section, the performance of the proposed joint resource allocation and power control scheme is evaluated via simulations. The EE per Hz is considered in the simulation.  The corresponding simulation parameters  \cite{Feng} are shown as follows: the cell radius is $500$m,
the pathloss constant $\kappa=10^{-2}$, the pathloss exponent $\chi = 4$, %the bandwidth of one RB $W = 180$kHz,
the standard deviation of
the shadowing is 8dB, the noise spectral density is -174dBm/Hz, $P_{\max} =  0.1$W, $N_c = 16$, $M = 16$, $\varepsilon = 0.0001$.
{In the following, for description convenience, we set the average received powers at the BS from the CUs to be the same, i.e., $P_{k,B} = P_B, \forall k \in \mathbb{L}$.
In a real scenario, different users may have different traffics and different QoS, and cell edge users in general are assigned lower target received powers.
However, we know that the locations of  CUs in the cell are generally different, and that the distances and the channel conditions between  CUs and the BS are almost always different.
According to (1), the transmit powers of the CUs, $P_{k,j}$, are generally different no matter whether the average received powers at the BS, $P_{k,B}$, are at the same level or not.
Correspondingly, the interference from  CUs towards  D2D users is different, too.
Therefore, whether the received power levels at the BS from  CUs are the same or not has no impact on the proposed joint resource allocation and power control scheme.}

In Figure \ref{fig2}, we show the performance comparison between the proposed scheme and the schemes in \cite{Yu1} and \cite{Xiao}. The scheme in \cite{Yu1} is rate adaptive and the scheme in \cite{Xiao} is margin adaptive. In \cite{Yu1}, maximizing throughput is equivalent to solve the optimal solution of a binary function. In \cite{Xiao}, the next bit is allocated to the subcarrier whose power increment in each loop is minimal.
Studying the curves in Figure \ref{fig2}, we can observe that the proposed scheme, which optimizes the measure used,  can greatly improve the EE of D2D communications  compared with the schemes in \cite{Yu1} and \cite{Xiao}, which optimize different measures.
Note here that the complexity of the proposed scheme is $\mathcal {O} (N_d M)$, and the complexities of the schemes in  \cite{Yu1} and \cite{Xiao}
are $\mathcal {O}  (N_d )$ and $\mathcal {O}  (M\sum\limits_{i = 1}^{{N_d}} {{R_i}})$, where $R_i$ denotes the required bits to be transmitted by the $i$-th D2D pair. We can see that our proposed scheme has polynomial complexity in regard to problem
scale $N_d$ and $M$, which facilitates the practical implementation.
It can also be observed from the figure that the EE of all the considered schemes decreases slightly with  the increase of $P_B/N_0$, which is due to the increased interference power from  CUs to D2D pairs. Another interesting observation is that the value of $P_B/N_0$  affects the EE to a small degree, and the reason is that   the allowed interference
from the considered CU sharing the same resource with the corresponding D2D pair is constrained to be smaller than a certain threshold.

\begin{figure}[!b]
\centering %\vspace*{135pt}
\includegraphics[width=0.45\textwidth]{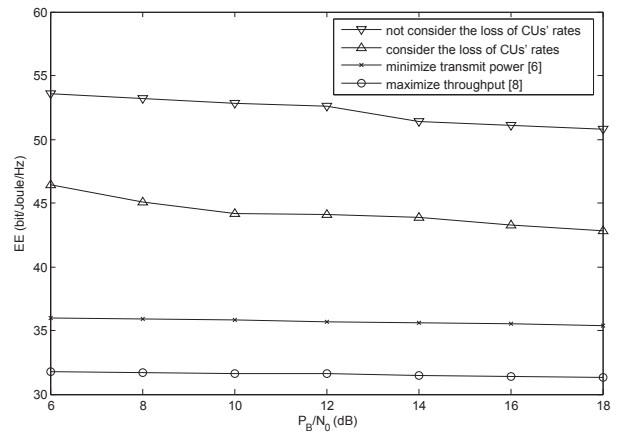}
\captionstyle{mystyle2}
\caption{EE versus $P_B/N_0$ with the considered different schemes.}
\label{fig2}
\end{figure}

\begin{figure}[!t]
\centering %\vspace*{135pt}
\includegraphics[width=0.45\textwidth]{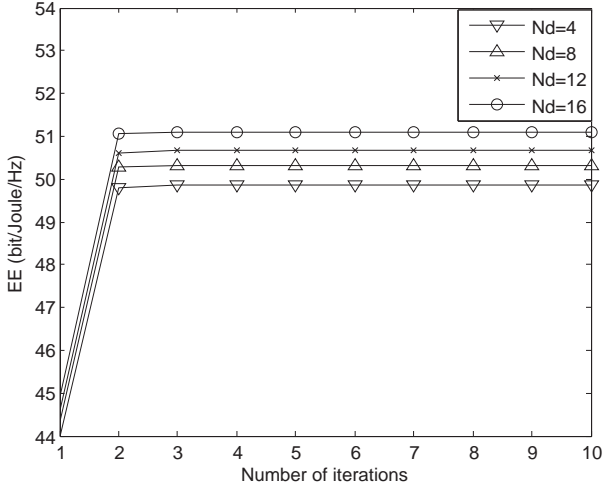}
\captionstyle{mystyle2}
\caption{EE versus number of iterations with different $N_d$.}
\label{fig3}
\end{figure}

In Figure \ref{fig3}, the EE of the proposed scheme versus the number of iterations with different $N_d$ is illustrated. It can be observed that the EE increases with the number of D2D pairs. It is obvious that the larger the number of  D2D pairs, the larger the EE of the proposed scheme.
It can also be observed that only 2 iterations are required to converge to the optimal EE for the considered simulation scenario.

In Figure \ref{fig4}, we show how the distance between the users in one D2D pair affects the EE  of  the proposed scheme. It can be observed that the EE decreases with the increasing of the distance between the D2D users. The reason is that the fading increases with the distance between the D2D users. As the distance between the  users in one D2D pair increases from 10m to 100m, the EE decreases about 66.7\%. It is obvious that the distance between the D2D users has a great influence on the performance of the proposed scheme.

In Figure \ref{fig5}, we illustrate the EE of the proposed scheme versus $\gamma_i$ (the minimum rate requirement of each D2D pair) with different $N_d$. It can be observed that there exists a saturation point beyond which the EE no longer increases with $\gamma_i$, and it shows that the minimum rate requirement influences the optimal transmit power value. This observation is helpful for the optimal energy efficient designs. We would like to remark here that the relationship between the minimum rate requirement and the EE {of D2D users} is coincident with that between the transmission rate and the EE {of D2D users} for a given system.
There is a certain relationship between the transmission rate and the EE {of D2D users}. The EE {of D2D users} is defined as the ratio of the transmission rate of all D2D users and the total energy consumption. The transmission rate {of D2D users} is the sum of {$\log$} functions, and the incremental rate of {$\log$} function decreases with the increase of the {D2D} transmit power. However, the incremental rate of the energy consumption is unchanged with the {D2D} transmit power.
{Correspondingly, when the incremental rate of  {$\log$} function is greater than that of the energy consumption, the EE of D2D users increases with the D2D transmit power.
Otherwise, the EE of D2D users decreases with the D2D transmit power.
Therefore, the EE of D2D users first increases then decreases with the D2D transmit power and also the transmission rate of D2D users.}
{We can see from the third term in the right hand side of (12) that the transmission rate of the considered CU by considering the influence of the corresponding D2D pair obviously decreases  with the D2D transmit power.
Combined with the above analytical results, the relationship between the EE of D2D users and the transmission rate of CUs can be readily established.
Clearly, there exists a tradeoff between the EE of D2D users and the transmission rate of CUs.
If we focus on the EE of D2D users, the optimum $P_{i,j}$ can be chosen, which incurs certain rate loss to  CUs.
On the other hand, if we give priority to the transmission rate of CUs, $P_{i,j}$ should be selected to be as small as possible, which may result in a lower EE for  D2D users.
}

\begin{figure}[!t]
\centering %\vspace*{135pt}
\includegraphics[width=0.45\textwidth]{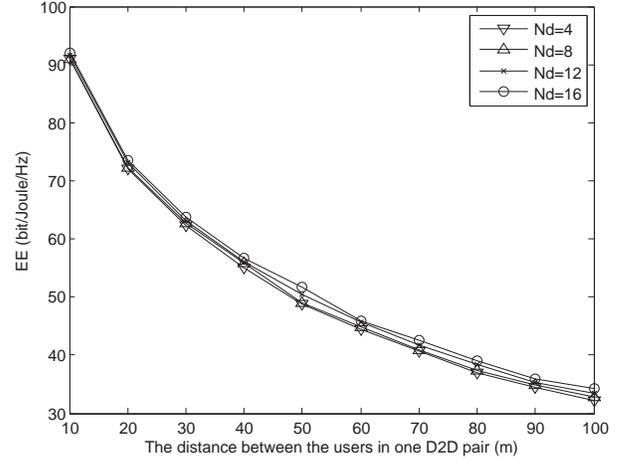}
\captionstyle{mystyle2}
\caption{EE versus the distance between the users in one D2D pair with different $N_d$.}
\label{fig4}
\end{figure}

\begin{figure}[!t]
\centering %\vspace*{135pt}
\includegraphics[width=0.45\textwidth]{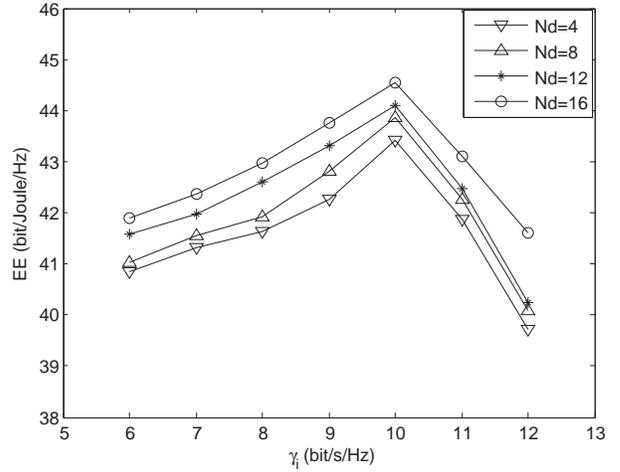}
\captionstyle{mystyle2}
\caption{EE versus $\gamma_i$ with different $N_d$.}
\label{fig5}
\end{figure}

\section{Conclusions}

In this paper, we have formulated the resource allocation and power control for energy efficient D2D communications
underlaying cellular networks as a  nonconvex optimization problem.
By exploiting the properties of fractional programming and penalty function,
an efficient iterative joint resource allocation and power control scheme has been derived to maximize the EE of D2D communications.
Simulation results have shown that the proposed scheme %revise 2014 %converges fast within a few iterations and
achieves remarkable improvements in terms of  EE. %demonstrated the trade-off between EE and SE.

%\IEEEtriggeratref{6}
\bibliographystyle{IEEEtran}
%\bibliography{manuscript_tvt}

\vspace*{-15pt}
\vspace*{-2\baselineskip}

\begin{biography}[{\includegraphics[width=1in,height
=1.25in,clip,keepaspectratio]{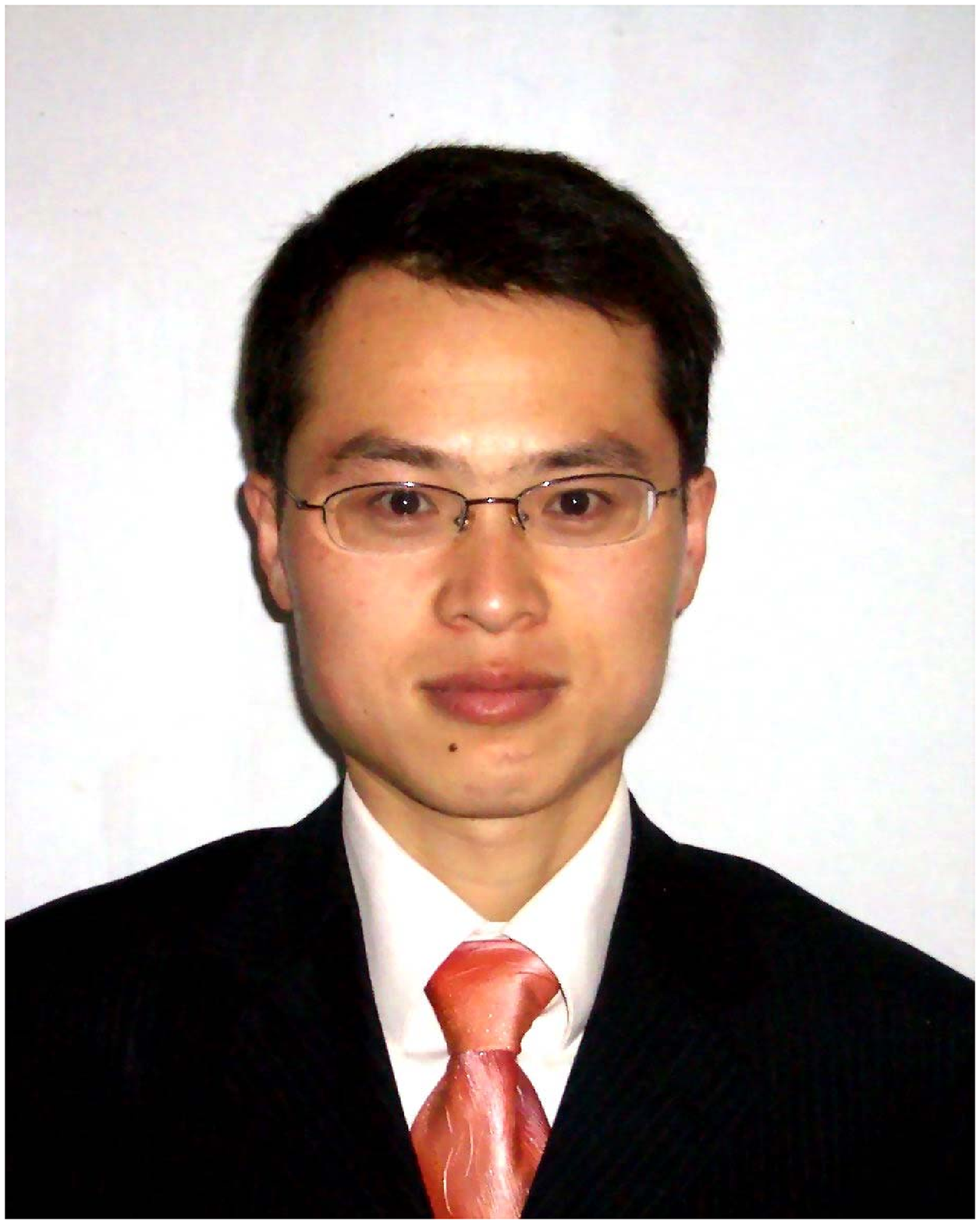}}]
{Yanxiang Jiang (S'03-M'07)}
received the B.S. degree in electrical engineering from Nanjing University, Nanjing, China, in 1999,
and the M.S. and Ph.D. degrees in communications and information systems from Southeast University, Nanjing, China, in 2003 and 2007, respectively.
From December 2013 to December 2014, he was a Visiting Scholar at the Department of Electrical and Computer Engineering,
University of Maryland, College Park, MD, USA.
He is currently with the faculty of the National Mobile Communications Research Laboratory at Southeast University, China.
His research interests include 5G mobile communication systems, massive MIMO systems, and green communication systems.
\end{biography}

\vspace*{-2\baselineskip}

\begin{biography}[{\includegraphics[width=1in,height
=1.25in,clip,keepaspectratio]{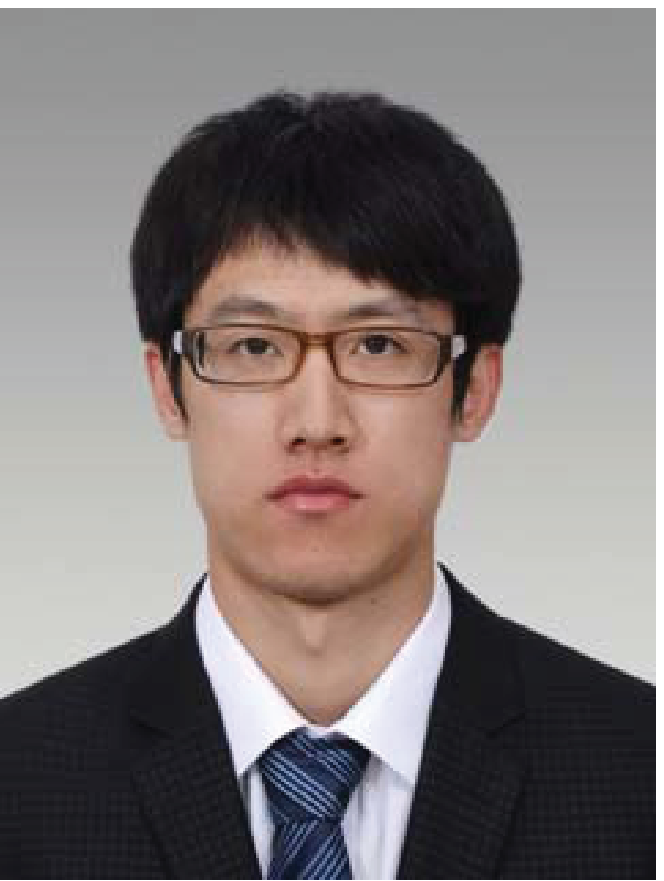}}]
{Qiang Liu}
received the M.S. degree in communications and information systems from Southeast University, Nanjing, China, in 2014.
He joined the Tangshan Branch, Hebei Co., Ltd., China Mobile Communications Group, in 2014.
His research interests include device-to-device systems, and mobile communication systems.
\end{biography}

\vspace*{-2\baselineskip}

\begin{biography}[{\includegraphics[width=1in,height
=1.25in,clip,keepaspectratio]{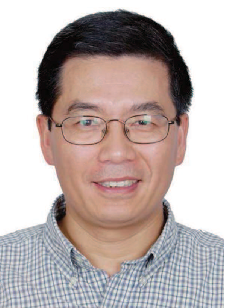}}]
{Fuchun Zheng (M'95-SM'99)}
obtained the BEng (1985) and MEng (1988) degrees in radio engineering
from Harbin Institute of Technology, China, and
the PhD degree in Electrical Engineering from the
University of Edinburgh, UK, in 1992.
From 1992 to 1995, he was a post-doctoral research associate with
the University of Bradford, UK, Between May 1995
and August 2007, he was with Victoria University,
Melbourne, Australia, first as a lecturer and then
as an associate professor in mobile communications.
He joined the University of Reading, UK,
in September 2007 as Professor (Chair) of Signal Processing. He has been
awarded two UK EPSRC Visiting Fellowships-both hosted by the University
of York (UK): first from August 2002 to July 2003 and then from August
2006 to July 2007. Over the past 15 years, Dr. Zheng has also carried out
and managed many industry-sponsored projects. He has been both a short
term visiting fellow and a long term visiting research fellow with British
Telecom, UK. Dr. Zheng's current research interests include signal processing
for communications, multiple antenna systems, and green communications.
He has been an active IEEE member since 1995. He was an editor (2001-
2004) of IEEE Transactions on Wireless Communications. In 2006, Dr. Zheng
served as the general chair of IEEE VTC 2006-Spring in Melbourne, Australia
(http://ieeevtc.org/vtc2006spring/)-the first ever VTC held in the southern
hemisphere. He received a VTC Chair Award from the IEEE VT Society
at IEEE VTC 2009-S (Barcelona, Spain) in April 2009. He will be the TPC
chair for VTC 2016-S in Nanjing (the first VTC in mainland China).

\end{biography}

\vspace*{-2\baselineskip}

\begin{biography}[{\includegraphics[width=1in,height
=1.25in,clip,keepaspectratio]{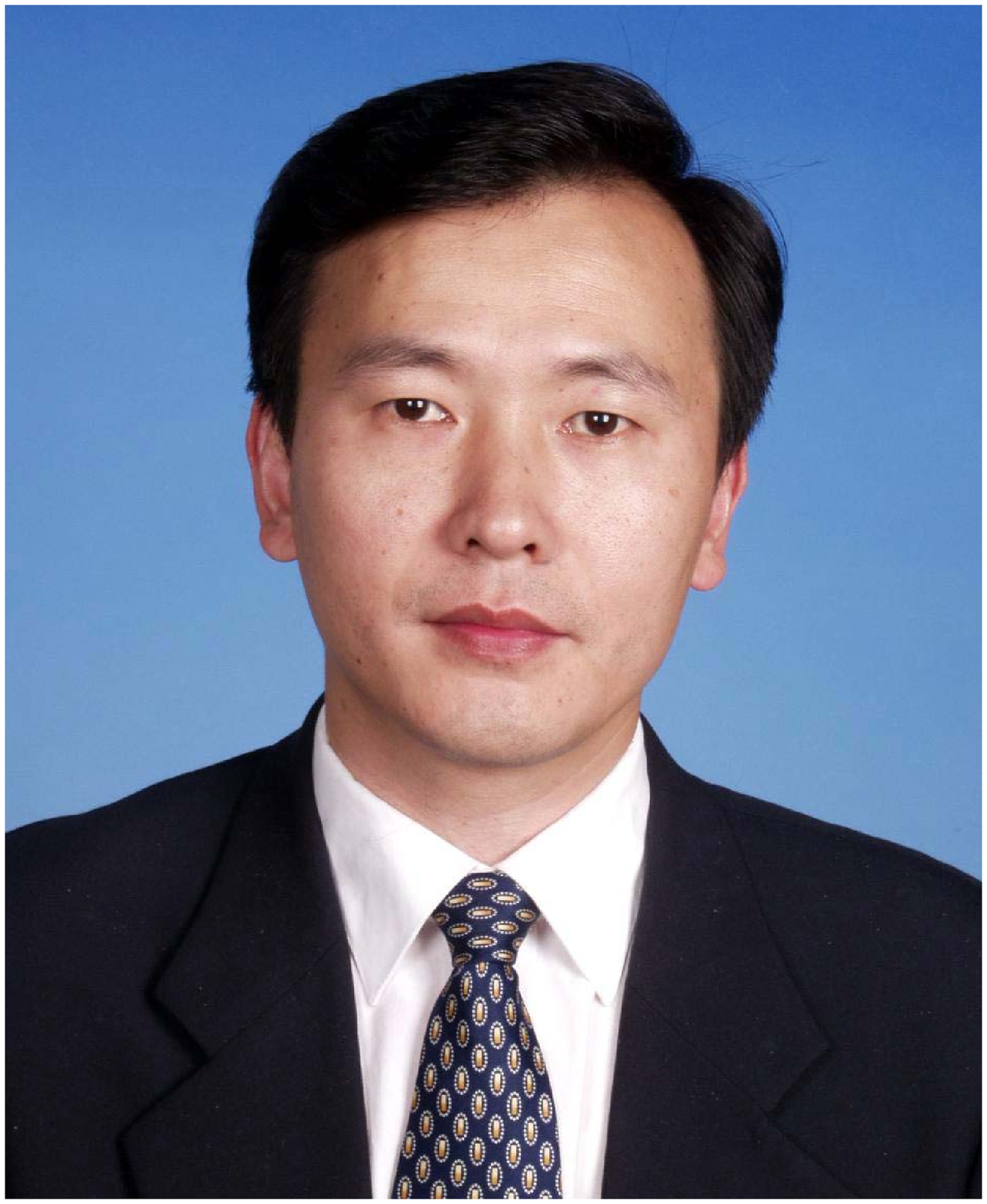}}]
{Xiqi Gao (SM'07-F'15)} received the Ph.D. degree
in electrical engineering from Southeast University,
Nanjing, China, in 1997. In April 1992, he joined the
Department of Radio Engineering, Southeast University,
where he has been a Professor of information
systems and communications since May 2001. From
September 1999 to August 2000, he was a Visiting
Scholar at Massachusetts Institute of Technology,
Cambridge, MA, USA, and Boston University,
Boston, MA. From August 2007 to July 2008, he visited
Darmstadt University of Technology, Darmstadt,
Germany, as a Humboldt Scholar. His current research interests include broadband
multicarrier communications, MIMO wireless communications, channel
estimation and turbo equalization, and multirate signal processing for wireless
communications.

Dr. Gao served as an Editor of the IEEE TRANSACTIONS ON WIRELESS
COMMUNICATIONS from 2007 to 2012. From 2009 to 2013, he served as an
Editor of the IEEE TRANSACTIONS ON SIGNAL PROCESSING. He is currently
serving as an Editor of the IEEE TRANSACTIONS ON COMMUNICATIONS. He
received the Science and Technology Awards of the State Education Ministry
of China in 1998, 2006, and 2009; the National Technological Invention
Award of China in 2011; and the 2011 IEEE Communications Society Stephen
O. Rice Prize Paper Award in the field of communications theory.

\end{biography}

\vspace*{-2\baselineskip}

\begin{biography}[{\includegraphics[width=1in,height
=1.25in,clip,keepaspectratio]{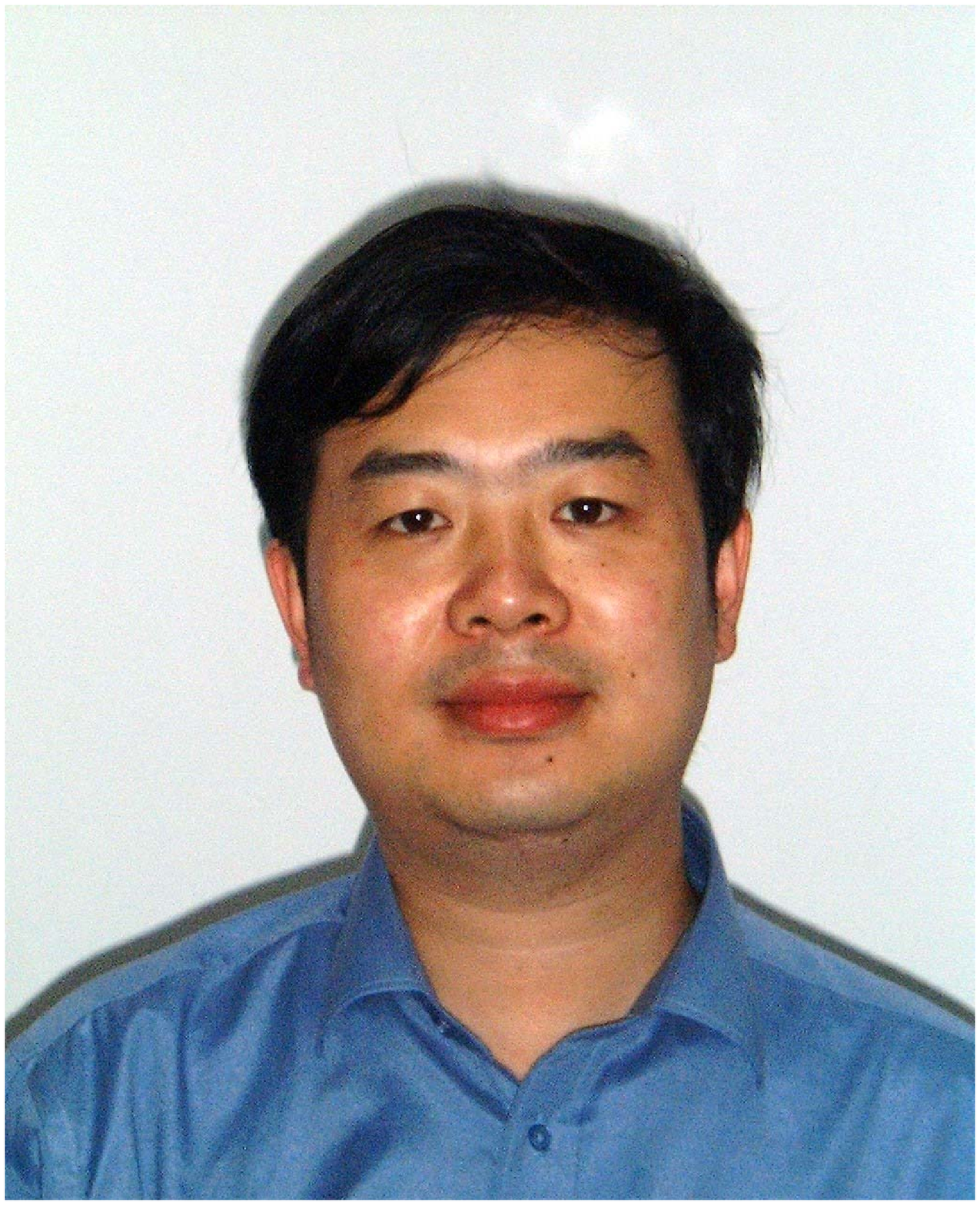}}]
{Xiaohu You (SM'11-F'12)} received the B.S., M.S., and Ph.D. degrees in electrical engineering
from Southeast University, Nanjing, China, in 1982, 1985, and 1988,
respectively. Since 1990, he has been working with National Mobile Communications
Research Laboratory at Southeast University, where he holds the
ranks of professor and director. He is the
Chief of the Technical Group of China 3G/B3G
Mobile Communication R$\&$D Project. His research
interests include mobile communications, adaptive
signal processing, and artificial neural networks, with applications to communications
and biomedical engineering.

Dr. You was a recipient of the Excellent Paper Prize from the China Institute
of Communications in 1987; the Elite Outstanding Young Teacher award from
Southeast University in 1990, 1991, and 1993; and the National Technological
Invention Award of China in 2011. He was also a recipient of the 1989 Young
Teacher Award of Fok Ying Tung Education Foundation, State Education
Commission of China.

\end{biography}

\end{document}